\documentclass[runningheads]{llncs}
\usepackage[utf8]{inputenc}

\usepackage[hidelinks]{hyperref}

\usepackage[T1]{fontenc}
\usepackage[USenglish]{babel}
\usepackage{amsmath}
\usepackage{amssymb}
\usepackage{mathtools}
\usepackage{listings}
\usepackage{graphicx}
\usepackage[noabbrev]{cleveref} 
\usepackage{algorithm}
\usepackage{enumitem} 

\let\subset\subseteq
\DeclareMathOperator{\inc}{inc} 
\DeclareMathOperator{\pos}{pos}
\DeclareMathOperator{\tig}{tig}
\DeclareMathOperator{\cover}{cr} 
 
\newcommand{\Cset}{\mathcal{C}}

\newcommand{\Cbar}{\mathcal{C}^{-1}}

\makeatletter
\def\blfootnote{\xdef\@thefnmark{}\@footnotetext}
\makeatother

\let\cref\Cref
\let\subset\subseteq

\usepackage{tikz}
\usepackage{tikzscale}
\tikzstyle{every node}=[draw,fill=white,shape=circle, inner sep = 1pt]

\begin{document}
\title{Drawing Order Diagrams Through Two-Dimension Extension}
%
%
\author{Dominik Dürrschnabel\inst{2,3} \and Tom Hanika\inst{1,2,3} \and Gerd
  Stumme\inst{2,3}}
%
%
\institute{
  Berlin School of Library and Information Science,\\ Humboldt University of
  Berlin, Berlin, Germany\\
  \and 
  Knowledge \& Data Engineering Group,
  University of Kassel, Kassel,  Germany\\
   \and
  Interdisciplinary Research Center for Information System Design\\
  University of Kassel, Kassel, Germany\\
  \email{duerrschnabel@cs.uni-kassel.de, hanika@cs.uni-kassel.de,
    stumme@cs.uni-kassel.de}
  }
\maketitle              
\blfootnote{Authors are given in alphabetical order.
  No priority in authorship is implied.}

\begin{abstract}
  Order diagrams are an important tool to visualize the complex structure of
  ordered sets. Favorable drawings of order diagrams, i.e., easily readable for
  humans, are hard to come by, even for small ordered sets. Many attempts were
  made to transfer classical graph drawing approaches to order
  diagrams. Although these methods produce satisfying results for some ordered
  sets, they unfortunately perform poorly in general. In this work we present
  the novel algorithm \texttt{DimDraw} to draw order diagrams. This algorithm is
  based on a relation between the dimension of an ordered set and the
  bipartiteness of a corresponding graph.
  
\keywords{Ordered~Sets \and Order~Diagrams \and Diagram~Drawing \and \mbox{Lattices}.}
\end{abstract}

\section{Introduction}\label{sec:intro}

\emph{Order diagrams}, also called \emph{line diagrams} or \emph{Hasse
  diagrams}, are a great tool for visualizing the underlying structure of
ordered sets. In particular they enable the reader to explore and interpret
complex information. In such diagrams every element is visualized by a point on
the plane. Each edge of the covering relation is visualized as an ascending line
connecting its points. These lines are not allowed to touch other points. These
strong requirements are often complemented with further soft conditions to
improve the readability of diagrams. For example, minimizing the number of
crossing lines or the number of different slopes. Another desirable condition is
to draw as many chains as possible on straight lines. Lastly, the distance of
points to (non-incident) lines should be maximized.

Experience shows that in order to obtain (human) readable drawings one has to
balance those criteria. Based on this notion, there are algorithms that produce
drawings of order diagrams optimizing towards some of the criteria mentioned
above. Drawings produced by such algorithms are sufficient to some
extent. However, they may not compete with those created manually by an
experienced human. However, such an expert is often not available, too
expensive, or not efficient enough to create a large number of order
diagrams. Hence, finding efficient algorithms that draw diagrams at a suitable
quality is still an open task. An exemplary requiring such algorithms is Formal
Concept Analysis (FCA) \cite{fca-book}, a theory that can be used to analyze and
cluster data through ordering it.

In this work we present a novel approach which does not employ the optimization
techniques as described above. For this we make use of the structure and its
properties that are already encapsulated in the ordered set. We base our idea on
the observation that ordered sets of order dimension two can be embedded into
the plane in a natural way. Building up on this we show a procedure to embed the
ordered sets of order dimension three and above by reducing them to the
two-dimensional case. To this end we prove an essential fact about
inclusion-maximal bipartite induced subgraphs in this realm. Based on this we
link the naturally emerging $\mathcal{NP}$-hard computation problem to a
formulation as a SAT problem instance. Our main contribution with respect to
this is \Cref{thm:main}.

We investigate our theoretical result on different real-world data sets using
the just introduced algorithm \texttt{DimDraw}. Furthermore, we note how to
incorporate heuristical approaches replacing the SAT solver for faster
computations. Finally, we discuss in part surprising observations and formulate
open questions.


\section{Related Work}
\label{sec:relatedwork}

Order diagrams can be considered as acyclic (intransitive) digraphs that are
drawn upward in the plane, i.e., every arc is a curve monotonically increasing
in $y$-direction. A lot of research has been conducted for such upward
drawings. A frequently employed algorithm-framework to draw such graphs is known
as Sugiyama Framework \cite{Sug81}. This algorithm first divides the set of
vertices of a graph into different layers, then embeds each layer on the same
$y$-coordinate and minimizes crossings between consecutive layers. Crossing
minimization can be a fundamental aesthetic for upward drawings. However the
underlying decision problem is known to be $\mathcal{NP}$-hard even for the case
of two-layered graphs \cite{eades1986edge}. A heuristic for crossing reduction
can be found in~\cite{eades1986median}. The special case for drawing rooted
trees can be solved using divide-and-conquer algorithms
\cite{journals/tse/ReingoldT81}. Such divide-and-conquer strategies can also be
used for non-trees as shown in~\cite{journals/tsmc/MessingerRH91}.  Several
algorithms were developed to work directly on order diagrams. Relevant for our
work is the dominance drawing approach. There, comparable elements of the order
relation are placed such that both Cartesian coordinates of one element are
greater than the ones of the other \cite{kelly1975planar}. Weak dominance
drawings allow a certain number of elements that are placed as if they were
comparable \cite{journals/corr/abs-1108-1439} even when they aren't. Our
approach is based on this idea. Previous attempts to develop heuristics are
described in~\cite{conf/gd/KornaropoulosT12a}.  If an ordered set is a lattice
there are algorithms that make use of the structure provided by this. The
authors in~\cite{stumme95geometrical} make use of geometrical representations
for drawings of lattices. In~\cite{freese2004automated} a force directed
approach is employed, together with a rank function to guarantee the ``upward
property'' is preserved. A focus on additive diagrams is laid out in
\cite{ganteradd}.

\section{Notations and Definitions}
\label{sec:notations}
We start by recollecting notations and notions from order theory
\cite{trotter1992combinatorics}. In this work we call a pair $(X,R)$ an
\emph{ordered set}, if $R\subset X \times X$ is an \emph{order relation} on a
set $X$, i.e., $R$ is reflexive, antisymmetric and transitive. In this setting
$X$ is called the \emph{ground set} of $(X,R)$. In some cases, we write
$(X,{\leq})$ instead of $(X,R)$ throughout this paper. We then use the
notations $(a,b)\in {\leq}$, $a \leq b$ and $b \geq a$ interchangeably. We write
$a<b$ iff $a \leq b$ and $a \neq b$. Alike, if $b \geq a$ and $b \neq a$, write
$b > a$. We say that a pair $(a,b)\in X \times X$ is \emph{comparable}, if
$a \leq b$ or $b \leq a$, otherwise it is \emph{incomparable}. An order relation
on $X$ is called \emph{linear} (or total) if all elements of $X$ are pairwise
comparable. For $(X,{\leq})$ the order relation $L$ on $X$ is called a
\emph{linear extension} of $\leq$, iff $L$ is a linear order and
${\leq} \subset L$. If $\mathcal{R}$ is a family of linear extensions of $\leq$
and ${\leq}=\bigcap_{ L \subset \mathcal{R}}L$, we call $\mathcal{R}$
\emph{realizer} of $\leq$. The minimal $d$ such that there is a realizer of
cardinality $d$ for  $(X,{\leq})$ is called its \emph{order
  dimension}. We use the denotation of order dimension for ordered sets and
order relations interchangeably. For a set $X$ and 
$\mathcal{C}\subset X \times X$ we denote
$\Cbar:=\{(a,b) \mid (b,a)\in \mathcal{C}\}$. The \emph{transitive closure} of
$R$ is denoted by $R^+$.

For our work we consider simple graphs denoted by $(V,E)$, where
$E \subset\binom{V}{2}$. For an ordered set $(X,{\leq})$, its \emph{comparability
  graph} is defined as the graph $(X,E)$, such that $\{a,b\}\in E$, if and only
if $a,b\in X$ are comparable. Similary the \emph{cocomparability graph}
(sometimes called incomparability graph) is the graph on $X$ where $\{a,b\}$ is
an edge if and only if $a,b\in X$ are incomparable. Two order relations on the
same ground set are called \emph{conjugate} to each other, if the comparability
graph of one is the cocomparability graph of the other.  We refrain from a
formal definition of a \emph{drawing} of $(X,{\leq})$. However we need to
discuss the elements of such drawings used in our work. Each element of the
ground set is drawn as a point on the plane. The \emph{cover relation} is
defined as
${\cover(X,{\leq})\coloneqq\{(a,b)\in {\leq}\mid \nexists c \in X : a < c <
  b\}}$. Each element of the cover relation is drawn as a monotonically
increasing curve connecting the points.

From here on some definitions are less common. For $(X,{\leq})$ we denote the
set of incomparable elements by $\inc(X,{\leq})$. Two elements
$(a,b),(c,d)\in \inc(X,{\leq})$ are called \emph{incompatible}, if their addition
to $\leq$ creates a cycle in the emerging relation, i.e., if there is some
sequence of elements $c_1,\ldots ,c_n\in X$, such that each pair
${(c_i,c_{i+1})\in{\leq}\cup(a,b)\cup(c,d)}$ with $i \in\{1,\ldots,n\}$ and
$c_{n+1}=c_1$. We call the graph $((\inc(X,{\leq}),E)$ with
$\{(a,b),(c,d)\}\in E$, iff $(a,b)$ and $(c,d)$ are incompatible the
\emph{transitive incompatibility graph}. Denote this graph by $\tig(X,{\leq})$.
We say a pair $(a,b)\in\inc(X,{\leq})$ \emph{enforces} another pair
$(c,d)\in\inc(X,{\leq})$, iff $(c,d)\in ({\leq}\cup(a,b))^+$. If and $(a,b)$
enforces $(c,d)$ in $\leq$, we write $(a,b)\rightarrow(c,d)$.

\section{Drawing Ordered Sets of Dimension Two}
\label{sec:2d}
Ordered sets of order dimension 2 have a natural way to be visualized using a
realizer by their dominance drawings \cite{kelly1975planar}. Let $(X,{\leq})$ be
an ordered set of dimension two. First define the \emph{position} for each $x\in X$ in a linear extension $L$ as the number of vertices that are
smaller, i.e., $\pos_{L}(x)\coloneqq\{y\in X\mid y < x\}$.

Now let $\mathcal{R}=|\{{\leq_1},{\leq_2}\}|$ be a realizer consisting of two
linear extensions of $\leq$.  Each element is embedded into a two-dimensional
grid at the pair of coordinates $(\pos_{\leq_1}(x),\pos_{\leq_2}(x))$. Embedding
this grid into the plane is done using the generating vector $(-1,1)$ for $x_1$
and $(1,1)$ for $x_2$. 
%
%
Each point now divides the plane into four quadrants using the two
lines that are parallel to $x_1$ and $x_2$. It holds that $a < b$, if and only
if the point $b$ is in the quadrant above the point $a$ by construction. Draw
the elements of the cover relation as straight lines. This guarantees that all
elements of the cover relation are drawn as monotonically increasing curves.
In order to compute such drawings, a preliminary check of the two-dimensionality
of the ordered set is required. If so computing a realizer in polynomial time is
possible as a result of the following theorem.

\begin{theorem}[Dushnik and Miller, 1941 \cite{dushnik1941partially}]
\label{thm:dushnikmiller}
The dimension of an ordered set $(X,{\leq})$ is at most 2, if and only if there is
a conjugate order $\leq_C$ on $X$. A realizer of $P$ is given by $\leq$ given by
$\mathcal{R}=\{{\leq}\cup{\leq_C},{\leq}\cup{\geq_C}\}$.
\end{theorem}

\noindent In 1977, Golumbic gave an algorithm \cite{golumbic1977complexity} to
check whether a graph is transitive orientable, i.e., whether there is an order
on its vertices, such that the graph is exactly the comparability graph of this
order. It computes such an order, if it exists. This algorithm runs in
$\mathcal{O}(n^3)$, with $n$ being the number of vertices of the
graph. Combining this with \Cref{thm:dushnikmiller} provides an algorithm to
compute whether an ordered set is two-dimensional. Furthermore, it also returns a
realizer, in the case of two-dimensionality. For the sake of completeness, note that
there are faster algorithms (as fast as linear)
\cite{Mcconnell97linear-timetransitive}) for computing transitive
orientations. However those only work if the graph is actually transitive
orientable and return erroneous results otherwise.  On a final note, deciding
dimension $3$ or larger is known to be $\mathcal{NP}$-complete as shown by
Yannakis in 1982 \cite{yannakakis1982complexity}. This is in contrast to the
just stated fact about two-dimensional orders.

\section{Drawing Ordered Sets of Higher Dimensions}
\label{sec:3d}

The idea of embedding two-dimensional orders does generalize in a natural way to
higher dimensions, i.e., gives us a nice way to embed $n$-dimensional ordered
sets into $n$-dimensional space (Euclidean space) by using an $n$-dimensional
realizer. However, projecting an ordered set from a higher dimension into the
plane turns out to be hard. See, for example, the projections of an ordered set
in example \Cref{fig:project}.  For this reason our algorithm makes use of a
different method to compute drawings of order diagrams for higher dimensions.

\begin{figure}[t]
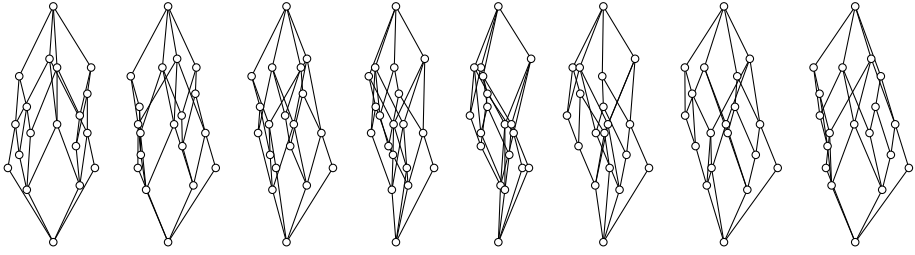

  \centering
  \includegraphics[height=10em]{img0}\hfill
  \includegraphics[height=10em]{img1}\hfill
  \includegraphics[height=10em]{img2}\hfill
  \includegraphics[height=10em]{img3}\hfill
  \includegraphics[height=10em]{img4}\hfill
  \includegraphics[height=10em]{img5}\hfill
  \includegraphics[height=10em]{img6}\hfill
  \includegraphics[height=10em]{img7}
  \caption{A three-dimensional ordered set embedded -- based on its realizer -- into
    three dimensional Euclidean space and then projected into the plane using a
    parallel projection from multiple angles. Even though the structure of the
    ordered set is recognizable, the drawings are all not satisfactory.}
  \label{fig:project}
\end{figure}

In short, the main idea of this section is the following: for a given order
relation we want to insert some number of additional pairs in order to make it
two-dimensional. This allows the resulting order to be drawn using the algorithm
for the two-dimensional case described in the previous section. Afterwards we
once again remove all the inserted pairs. By construction, the property that if
$a < b$, the point $a$ is inserted below $b$ is still preserved.  Such drawings
are sometimes called \emph{weak dominance drawings}
\cite{journals/corr/abs-1108-1439}. However, for each inserted pair $(a,b)$, we
obtain two points in the drawing that are drawn as if $a$ and $b$ were
comparable. This poses the question for minimizing the number of inserted pairs.

\begin{definition}
  Let $(X,{\leq})$ be an ordered set. A set $\mathcal{C}\subset \inc(X)$ is called
  a \emph{two-dimension-extension} of $(X,{\leq})$, iff ${\leq}\cup\mathcal{C}$ is
  an order on $X$ and the ordered set $(X,{\leq}\cup\mathcal{C})$ is
  two-dimensional.
\end{definition}
Such an extension always exists: a linear extension of dimension one
always exists. If an order contains exactly one incomparable pair it has
dimension two.

It is known to be $\mathcal{NP}$-complete to decide whether an ordered set can
be altered to be two-dimensional by inserting $k$ pairs \cite{Brightwell_2013}.
Hence, we propose an algorithm that tackles the problem for approximating the
corresponding optimization problem. The idea of the algorithm is based on the
following theorem.

\begin{theorem}[Doignon et.al., 1984 \cite{doignon1984realizable}]\label{thm:doignon}
  The ordered set $(X,{\leq})$ has order-dimension two if and only if
  $\tig(X,{\leq})$ is bipartite.
\end{theorem}

\noindent Thus, for $(X,{\leq})$ of dimension greater than two, 
$\tig(X,{\leq})$ is non-bipartite. We want to find a maximal induced
bipartite subgraph of $\tig(X,{\leq})$.  

\begin{lemma}\label{thm:ring}
  Let $(X,{\leq})$ be an ordered set and $(a,b),(c,d)\in\inc(X,{\leq})$. Then the
  following are equivalent:
  \begin{enumerate}[label=\roman*)]
  \item $d \leq a$ and $b \leq c$.
  \item $(a,b)\rightarrow(d,c)$.
  \item $(c,d)\rightarrow(b,a)$.
  \item $(a,b)$ and $(c,d)$ are incompatible.
  \item $(b,a)$ and $(d,c)$ are incompatible.
  \end{enumerate}
\end{lemma}

\begin{proof}
  $(i)\Rightarrow(iv)$. Consider the relation
  $\prec \coloneqq{\leq} \cup (a,b) \cup (c,d)$. This yields
  $d \prec a \prec b \prec c \prec d$, i.e, $\prec$ contains a
  cycle. Analogously $(i) \Rightarrow (v)$.\\%
  $(iv)\Rightarrow(ii)$. The assumption directly implies that
  $(d,c) \in ({\leq}\cup(a,b))^+$ due to the cycle generated by $(a,b)$ and
  $(c,d)$. By the same argument $(v)\Rightarrow (iii)$.\\%
  $(ii)\Rightarrow(i)$. Assume $d\not \leq a$ or $b \not \leq c$. Since
  $(d,c)\in({\leq}\cup(a,b))^+$ it follows $(d,c)\in {\leq}$ which contradicts
  $(c,d)\in \inc(X,\leq)$. Similarly follows $(iii)\Rightarrow (i)$.
  \null\hfill $\square$
\end{proof}

\noindent Recall the definition of $\tig(X,{\leq})$ defined on $\inc(X,{\leq})$
with incompatible pairs being connected. Call a cycle in
$\tig(X,{\leq})$  \emph{strict}, iff for each two adjacent pairs $(a,b)$ and
$(c,d)$ it holds that $d < a$ and $b < c$. A strict path is defined analogously.

\begin{lemma}[Doignon et al., \cite{doignon1984realizable}]
  \label{thm:strict}
  Let $(X,{\leq})$ be an ordered set and let the pair $v\in\inc(X,{\leq})$. Then the
  following statements are equivalent:
  \begin{enumerate}[label=\roman*)]
  \item $v$ is in contained in an odd cycle in $\tig(X,{\leq})$.
  \item $v$ is contained in a strict odd cycle in $\tig(X,{\leq})$.
  \end{enumerate}
\end{lemma}

\begin{remark}
  This is stated implicitly in their proof of Proposition 2, verifying the
  equality between the two chromatic numbers of a hypergraph corresponding to
  our cycles and a hypergraph corresponding to our strict cycles.
\end{remark}

\begin{theorem}\label{thm:main}
  Let $(X,{\leq})$ an ordered set. Let $\Cset\subset\inc(X,{\leq})$ be minimal with
  respect to set inclusion, such that $\tig(X,{\leq})\backslash\Cbar$ is
  bipartite. Then $(X,{\leq}\cup\Cset)$ is an ordered set.
\end{theorem}

\begin{proof}
  Refer to the bipartition elements of $\tig(X,{\leq})\backslash\Cbar$ with
  $P_1$ and $P_2$.

  \begin{claim}
    The arrow relation is transitive, i.e., if $(a,b)\rightarrow(c,d)$ and
    $(c,d) \rightarrow (e,f)$ then $(a,b)\rightarrow (e,f)$. If
    $(a,b)\rightarrow (c,d)$ then $c\leq a$ and $b\leq d$ by
    definition. Similarly $(c,d)\rightarrow (e,f)$ implies $e \leq c$ and
    $d \leq f$. By transitivity of $\leq$ this yields that $e\leq a$ and
    $b \leq f$ which in turn implies that $(a,b)\rightarrow(d,f)$.
  \end{claim}  

  \begin{claim}[$\star$]Let $(a,b),(c,d)\in\inc(X,{\leq})$ with
    $(a,b)\not \in \Cbar$ and $(a,b)\rightarrow (c,d)$, then
    $(c,d)\not \in \Cbar$. Assume the opposite, i.e., $(c,d) \in \Cbar$. Without
    loss of generality let $(a,b) \in P_1$. As $(c,d) \in \Cbar$ there has
    to be a pair $(e,f)\in P_1$ that is incompatible to $(c,d)$, i.e.,
    $(e,f)\rightarrow (d,c)$, otherwise $(c,d)$ can be added to $P_1$ without
    destroying the independet set. However $(a,b)\to (c,d)$ is
    equivalent to $(d,c)\rightarrow (b,a)$ and yields together with the transitivity
    of the arrow relation $(e,f)\rightarrow (b,a)$. But than $(e,f)$ and $(a,b)$
    are incompatible, a contradiction since both are in the independent set
    $P_1$.
  \end{claim}
  
  \begin{claim}[$\star \star$]
    If $(x,y)\in \Cset$, then $(y,x)\not\in\Cset$. As $(y,x)\in \Cbar$, there is
    a pair $(a,b)$ in $P_1$, such that $(a,b)$ and $(y,x)$ are incompatible,
    i.e., $(a,b)\rightarrow (x,y)$, otherwise $(y,x)$ can be added to
    $P_1$. However, since $(a,b)\not \in \Cbar$ follows $(x,y)\not \in \Cbar$ by
    Claim ($\star$).
  \end{claim}

  \noindent Reflexivity: $\forall x\in X$ we have
  $(x,x)\in {\leq}\subset {\leq}\cup\Cset$.

  \noindent Antisymmetry: assume $(x,y)\in {\leq}\cup\Cset$ and
  $(y,x)\in{\leq}\cup\Cset$. We have to consider three cases. First,
  $(x,y)\in{\leq}$ and $(y,x)\in{\leq}$. Then $x=y$, as $\leq$ is an order
  relation. %
  Secondly, $(x,y)\in{\leq}$ and $(y,x)\in\Cset$. If $(x,y)\in{\leq}$, then $x$
  and $y$ are comparable, i.e., the pair $(y,x)$ can't be in
  $\inc(P,{\leq})$. Then $(y,x)\not\in\Cset$, a contradiction. %
  Thirdly, $(x,y)\in \Cset$ and $(y,x)\in\Cset$. This may not occur by
  Claim~$(\star\star)$.

  \noindent Transitivity: let $(x,y)\in{\leq}\cup\Cset$ and
  $(y,z)\in{\leq}\cup\Cset$ we show $(y,z)\in{\leq}\cup\Cset$. We have to
  consider four cases. First, $(x,y)\in{\leq}$ and $(y,z)\in{\leq}$ implies
  $(x,z)\in{\leq}$. %
  Secondly, $(x,y)\in{\leq}$ and $(y,z)\in\Cset$ and assume that
  $(x,z)\not\in({\leq}\cup\Cset)$. Then $(z,x)\not \in \Cbar$, but
  $(z,x)\rightarrow (z,y)$, as $(x,y)\in{\leq}$. From Claim $(\star)$ follows
  that $(z,y)\not\in\Cbar$, a contradiction to $(y,z)\in \Cset$. The case
  $(x,y)\in \Cset$ and $(y,z)\in{\leq}$ is treated analogously. %
  Lastly, $(x,y) \in \Cset$ and $(y,z) \in \Cset$ and assume
  $(x,z)\not \in \Cset$, i.e., $(z,x)\in \Cbar$. There has to be an odd cycle in
  $P_1\cup P_2$ together with $(y,z)$, otherwise $(y,z)$ can be added to
  $P_1\cup P_2$ to create a larger biparite graph. By \Cref{thm:strict}, there
  also hast to be a strict odd cycle. Let the neighbors of $(y,z)$ in
  $P_1\cup P_2$ be $(a,b)$ and $(c,d)$. Then $a < z$, $c < z$, $y < b$ and
  $y < d$, and the pairs $(a,b)$ and $(c,d)$ are connected by a strict path on
  an even number of vertices through the strict odd cycle. By the same argument
  there are pairs $(e,f)$ and $(g,h)$ with $e < y$, $g < y$, $x < f$ and $x < h$
  and $(e,f)$ and $(g,h)$ connected by a strict odd path. We now show, that
  $(z,x)$ is in an odd cycle with $P_1\cup P_2$ to yield a contradiction. For
  this consider the following paths $A=(c,f)(z,x)(h,a)$, $B=(d,h)(h,g)$ and
  $C=(f,e)(e,b)$. Each of those is a path in $\tig(P,{\leq})$ by definition.

  \begin{claim}
    Between $(h,a)$ and $(d,h)$ there is a path on an even number of vertices in
    $P_1\cup P_2$. To show this, let $(a_1,b_1),\ldots ,(a_{2k},b_{2k})$ be the
    strict path on an even number of vertices connecting $(a,b)$ and $(d,c)$
    such that $(a_1,b_1)$ = $(a,b)$ and $(a_{2k},b_{2k})=(d,c)$. This implies
    $a_{2i+1}< b_{2i+2}$, $a_{2i}< b_{2i+1}$, $b_{2i+1}> a_{2i+2}$ and
    $a_{2i}> b_{2i+1}$ and for each $i \in \{0,\ldots ,k-1\}$. However, this
    yields the path $(h,a)=(h,a_1)(b_2,h)(h,a_3)\cdots(b_{2k},h)=(d,h)$ which is
    even and connecting $(h,a)$ and $(d,h)$ in $P_1\cup P_2$, as required.
  \end{claim}

  \noindent Analogously we obtain a path between $(c,f)$ and $(b,f)$ on an even
  number of vertices. Moreover $(h,g)$ and $(f,e)$ are also connected by a path
  on an even number of vertices in $P_1\cup P_2$, since $(g,h)$ and $(e,f)$ are
  connected by an even path. Reversing all pairs of this path yields the
  required path. Combining the segments $A$, $B$ and $C$ with the paths
  connecting them yields an odd cycle in $P_1\cup P_2\cup \{(z,x)\}$, a
  contradiction.  \null\hfill $\square$
\end{proof}

\subsection{The Importance of Inclusion-Maximality}
Consider the standard example $S_3=(X,{\leq})$ where the ground set is defined as
$X=\{a_1,a_2,a_3,b_1,b_2,b_3\}$ and $a_i\not\leq a_j$ for $i \neq j$,
$b_i \not \leq b_j$ for $i \neq j$ and $a_i\leq a_j$ if and only if $i \neq
j$. This example is well known to be a three-dimensional ordered set. However,
it becomes two-dimensional by inserting a single pair $(a_i,b_i)$ into the order
relation $\leq$ for some index $i\in\{1,2,3\}$, i.e., the transitive
incomparability graph becomes bipartite if we remove for example the pair
$(b_1,a_1)$. Now assume we do not require to removing a set minimal with respect
to set inclusion, take for example both pairs $(a_1,b_1)$ and
$(b_1,a_1)$. However the set $(X,{\leq}\cup\{(a_1,b_1),(b_1,a_1)\})$ is not an
ordered set, as both pairs $(a_1,b_1)$ and $(b_1,a_1)$ are in
${\leq}\cup\{(a_1,b_1),(b_1,a_1)\}$. This is a conflict with $a_1\neq b_1$,
i.e., we do not preserve the order property as the resulting relation is not
antisymmetric.

\subsection{Bipartite Subgraph is not Sufficient}
\label{sec:insuffic}
From \Cref{thm:main} one might conjecture that finding an inclusion-minimal
bipartite subgraph of $\tig(X,{\leq})$ is sufficient to find a two-dimensional
extension of $(X,{\leq})$. However, it may occur that two pairs are not
incompatible in $\tig(X,{\leq})$ and are incompatible in
$\tig(X,{\leq}\cup\Cset)$ with $\Cset$ being an inclusion-minimal set such that
$\tig(X,{\leq})\backslash\Cbar$ is bipartite. This can arise in particular, if
the following pattern occurs: the ordered set contains the elements $a,b,c$ and
$d$, such that all elements are pairwise incomparable, except $b < d$ and
$(c,a)\in \Cset$ exhibits this observation, see \Cref{fig:counterexample1}. Then
$(a,b)$ and $(c,d)$ are not incompatible in $\tig(X,{\leq})$, but they become
incompatible with the relation ${\leq}\cup\Cset$.  An example for this is provided by the
ordered set in \Cref{fig:counterexample2}.  The transitive incomparability graph
of this ordered set has 206 vertices. By removing pairs
 
\begin{figure}[t]\centering 
  \begin{minipage}[t]{0.4\textwidth} 
    \centering
    \includegraphics{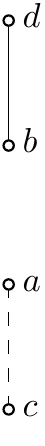}
    \caption{An example how new incompatibilities can arise. $\leq$ is the
      continous line, $\Cset$ is the dashed line. $(a,b)$ and $(c,d)$ are not
      incompatible in $\leq$ and incompatible in ${\leq} \cup \Cset$.}
    \label{fig:counterexample1}
  \end{minipage}
  \hfill
  \begin{minipage}[t]{0.55\textwidth} 
    \centering 
    \includegraphics{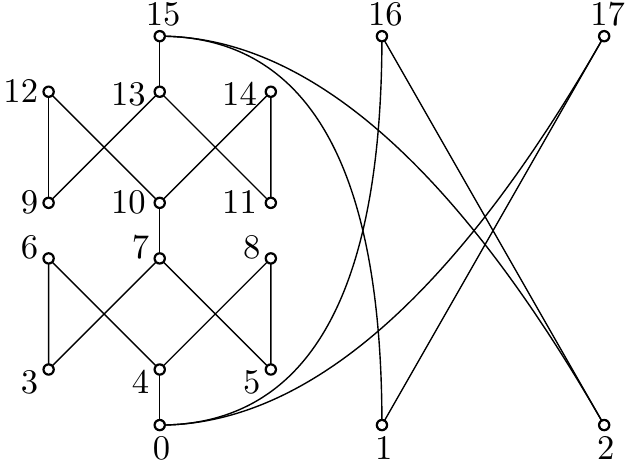}
    \caption{An example for an ordered set that has a transitive incompatibility
    graph with an inclusion-minimal bipartite subgraph of the transitive
    incompatibility graph that does not give rise to a two-dimension extension.}
    \label{fig:counterexample2}
  \end{minipage}
\end{figure}

\noindent \quad \texttt{\small (5,9), (3,13), (6,15), (17,15), (5,15), (5,13), (11,10),
  (9,6), (5,6), \\\null\quad(4,15), (1,8), (3,15), (11,12), (2,8), (11,7), (0,15), (1,16)},

\noindent the transitive incomparability graph of this ordered set becomes
bipartite. However, if we once again compute the transitive incomparability
graph of the new ordered set, we see that it is not bipartite, i.e., the new
ordered set is once again not two-dimensional. We have to add the additional
pair \texttt{(17,8)} to make the graph two-dimensional. It may be remarked at
this point that it is in fact possible to make the transitive incomparability
graph bipartite by adding only eleven pairs (in contrast to the seventeen added
in this particular example), see \Cref{fig:counterexample_alg}. Those pairs give
rise to a two-dimension extension.

\section{Algorithm}
\label{sec:alg}
Building up on the ideas and notions from the previous sections we propose the
Algorithm \texttt{DimDraw} as depicted in \Cref{alg:dimdraw}. Given an ordered
set $(X,{\leq})$, one calls \texttt{Compute\_Coordinates}. Until this procedure
identifies a conjugate order using the \texttt{Comupte\_Conjugate\_Order} (and
in turn the algorithm of Golumbic \cite{golumbic1977complexity}), it computes
bipartite subgraphs of the transitive incompatibility graphs. Furthermore it
adds the so-computed pairs to the $\leq$. For finite ground sets the algorithm
terminates after finitely many steps.

\begin{algorithm}[t]
  \caption{\texttt{DimDraw}}
  \label{alg:dimdraw}
  Execute \texttt{Compute\_Coordinates} on the ordered set that is to be drawn.
  \hrule
\begin{lstlisting}[mathescape=true]
Input: Ordered set ($P$,$\leq$)
Output: Conjugate order of ($P$,$\leq$)

def Compute_Conjugate_Order($P$,$\leq$):
    $C$ $\coloneqq$ Cocomparability_Graph($P$,$\leq$)
    if Has_Transitive_Orientation($C$):
        ($P$,$\leq_C$) $\coloneqq$ Transitive_Orientation($C$)
        return $\leq_C$
    else:
        return $\perp$
\end{lstlisting}
\hrule
\begin{lstlisting}[mathescape=true]
Input: Ordered set ($P$,$\leq$)
Output: Coordinates of the drawing of ($P$,$\leq$)

def Compute_Coordinates($P,\leq$):
    $\leq_C$ $\coloneqq$ Compute_Conjugate_Order($P,\leq)$
    $\mathcal{C}$ $\coloneqq$ $\emptyset$
    while $\leq_C$ $=$ $\perp$:
        $I$ $\coloneqq$ Transitive_Incomparability_Graph($P,{\leq}\cup\Cset$)
        $B$ $\coloneqq$ Maximum_Bipartite_Subgraph($I$)
        $\mathcal{C}$ $=$ $\mathcal{C}\cup V(I\backslash B)^{-1}$
        $\leq_C$ $\coloneqq$ Compute_Conjugate_Order($P$,${\leq} \cup \Cset$)
    $\leq_1$ $\coloneqq$ ${\leq}\cup{\leq_C}$
    $\leq_2$ $\coloneqq$ ${\leq}\cup{\geq_C}$
    for $x$ in $P$:
        Coord($x$,1) $\coloneqq$ $|\{k\mid k \leq_1 x\}|-1$
        Coord($x$,2) $\coloneqq$ $|\{k\mid k \leq_2 x\}|-1$
    return Coord
\end{lstlisting}
\end{algorithm}

\subsection{Postprocessing}
The algorithm does not prevent a point from being placed on top of lines
connecting two different points. This however is not allowed in order
diagrams. Possible strategies to deal with this problem are the following. One
strategy is to modify the coordinate system, such that the marks of different
integers are not equidistant. Another one is to perturb the points on lines
slightly. A third way is to use splines for drawing the line in order to avoid
crossing the point.

\section{Finding large induced bipartite subgraphs}
\label{sec:bip}
Our algorithm has to compute an inclusion-minimal set of vertices, such that
removing those vertices from the transitive incompatibility graph results in a bipartite
graph. Deciding for a graph whether it is possible to make it bipartite by
removing a set of cardinality $k$ is known to be $\mathcal{NP}$-complete
\cite{LEWIS1980219}. Even approximations are known to be in this complexity
class \cite{approx}. Therefore we propose different approaches.

\subsection{Exact solution using a reduction to SAT}

Even for small example, i.e., orders on less than 30 elements, a naive approach
is infeasible.  As we will see in \Cref{sec:evaluation} the question for
computing $\Cset$ results in $\binom{182}{5}$ tests for an example on $19$
elements (\cref{fig:fische}) and $\binom{294}{29}$ tests for example
(\cref{fig:car}) on 24 elements. Therefore we need a more sophisticated solution
for the problem. We reduce the problem for finding biparite subgraphs to a SAT
problem and then solve this problem with a SAT-Solver, in our case MiniSat
\cite{minisat} in version 2.2. In other words we want to know for some graph
$G=(V,E)$ on $n$ vertices and $m$ edges whether by deleting $k$ vertices we can
make the graph bipartite. Solving is done by finding a partition of $V$ into the
three sets $P_1,P_2,\Cbar$, such that $P_1$ and $P_2$ are independent sets and
$|\Cbar|\leq k$. For this we construct a conjugative normal form as follows: for
each vertex $v_i$ we have three variables, call them
$V_{i,1},V_{i,2},V_{i,3}$. The first two variables indicate, whether the vertex
is placed in $P_1$ or $P_2$, respectively and the third variable indicates,
whether the vertex is placed in $\Cbar$. For each vertex we have to guarantees,
that it is placed in at least one of $P_1$, $P_2$ or $\Cbar$, i.e., at least one
of $V_{i,1},V_{i,2},V_{i,3}$ is true for each $i$. We achieve this with the
clause $V_{i,1} \vee V_{i,2}\vee V_{i,3}$ for all $i \in \{1,\ldots,n\}$. Also
we want to guarantee that $P_1$ and $P_2$ are independent sets, i.e., no two
vertices in $P_1$ or $P_2$ are connected by an edge. We achive this by adding
the two clauses $\neg V_{i,1} \vee \neg V_{j,1}$ and
$\neg V_{i,2} \vee \neg V_{j,2}$ for each edge $\{v_i,v_j\}\in E$. This ensures
that no two vertices connected by an edge are placed in set $P_1$ or
$P_2$. Finally we have to guarantee, that there are at most $k$ vertices in
$P_3$, i.e., that at most $k$ of the variables
$\{V_{1,3},V_{2,3},\ldots,V_{n,3}\}$ are true. There are multiple ways to
achieve this. We employ the method as described in~\cite{conf/cp/Sinz05}. This
results in an additional $(n-1)\cdot k$ auxiliary variables and $2nk+n-3k-1$
clauses. Altogether our SAT instance has in total $(n-1)(k+3)+3$ variables and
$2m+2nk+2n-3k-1$ clauses. For this CNF the following holds by construction.
\begin{theorem}
  The SAT instance as constructed above is satisfiable if and only if $G$ has an
  induced bipartite subgraph on $n-k$ vertices.
\end{theorem}

\noindent Now we build this SAT instance for $k$ increasing from 1 until it is
satisfiable. Then the set $\{v_i \mid V_{i,3}=\text{true} \}$ is exactly the
subset of vertices we have to remove to make the graph bipartite. Obviously
methods like binary search may be applied here.  Furthermore, we may also plug
heuristic procedures into our algorithm to find an inclusion-minimal set
$\Cset$. We experimented on this with a greedy algorithm, a simulated annealing
approach and a genetic algorithm with promising results, especially for the
genetic algorithm. However, it is preferred to use the SAT algorithm as long as
there is enough computational power and the problem instances are not too large.

\section{Experimental evaluation}
\label{sec:evaluation}
\begin{figure}[t]
  \null\hfill
  \includegraphics[height=10em]{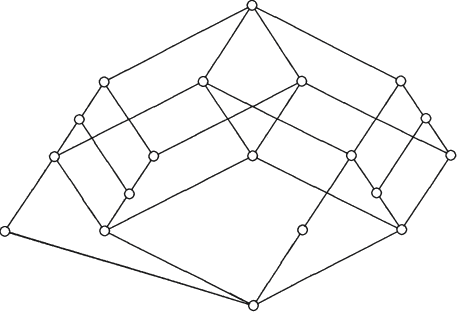}\hfill
  \includegraphics[height=10em]{fische}\hfill\null
  \caption{Two drawings of the ``Living Beings and Water'' lattice, by hand
    (left) and with our algorithm (right).}
  \label{fig:fische}
\end{figure}

\begin{figure}[t]
  \null\hfill
  \includegraphics[height=10em]{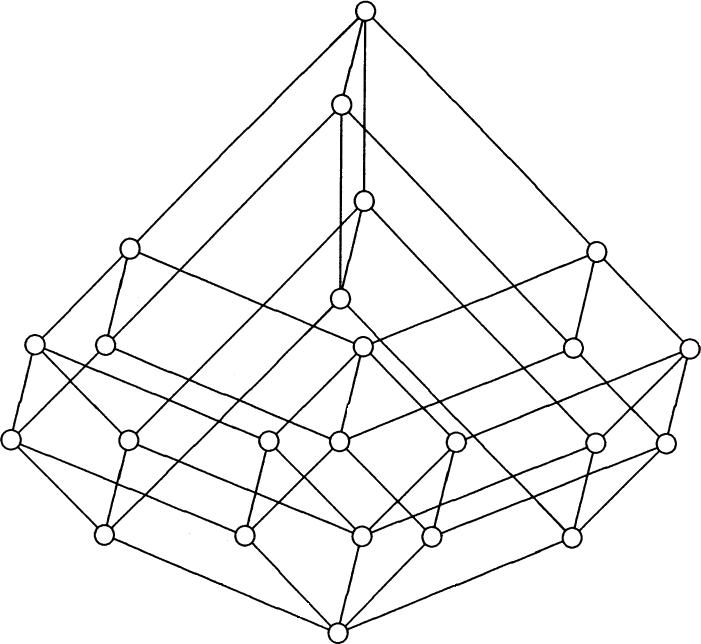} \hfill
  \includegraphics[height=10em]{car1} \hfill
  \includegraphics[height=10em]{car2} \hfill\null
  \caption{The ``Drive concepts for motorcars'' lattice. Drawn by an expert
    (left) and two drawings of our algorithm (middle and right).}
  \label{fig:car}
\end{figure}

The \texttt{DimDraw} algorithm was originally designed with the idea in mind to
draw the order diagram of lattices. Those are employed in a particular in Formal
Concept Analysis (FCA), a mathematical theory for analyzing data. Note that any
complete lattice can be represented by a concept lattice in FCA. We tested our
algorithm on all lattice examples from a standard literature book on FCA
\cite{fca-book}. In all those cases the quality of the produced drawings came
close to examples hand drawn by experts. For example, consider the lattice that
arises from ``Living Beings and Water'' \cite[p.18]{fca-book}. In
\Cref{fig:fische} we compare the hand-drawn example (left) to the result drawn
by our algorithm (right). For \Cref{fig:car} \cite[p.40]{fca-book} there are two
different solutions depicted, both having the minimal number of pairs
inserted, note that the algorithm stops after it finds a single solution.

Because of the importance of drawings in FCA we tested the algorithm on every
lattice with eleven or less vertices. The reader might want to have a look at the
document containing all 44994 drawings on 7499 pages
\cite{durrschnabel_dominik_2019_3242627}.

Concluding the experiments we want to present an example that our algorithm also
works on non-lattices. Consider the ordered set from
\Cref{fig:counterexample2}. While the hand-drawn version of this order diagram
makes use of splines, our algorithm-generated version
\Cref{fig:counterexample_alg} uses exclusively straight lines.

An interesting observation during the experiments was the following: for all
examples of that we are aware, even including those not presented in this work,
one pass of the SAT solver was sufficient for reducing the order dimension to
two. This is surprising, in particular in light of~\cref{fig:counterexample2}
from~\cref{sec:insuffic}.

\begin{figure}[t]
  \centering
  \includegraphics[height=10em]{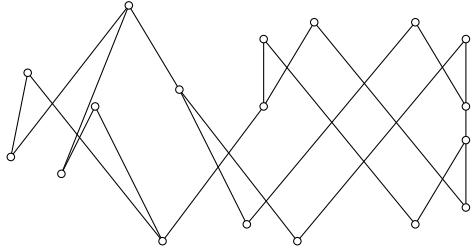}
  \caption{A non-lattice example drawing produced by the algorithm. See
    \Cref{fig:counterexample2} for a hand-drawn version.}
  \label{fig:counterexample_alg}
\end{figure}

\section{Conclusion and Outlook}
\label{sec:outlook}
We presented in this work a novel approach for drawing diagrams of order
relations. To this end we employed an idea by Doignon et al.\ relating order
dimension and bipartiteness of graphs and proved an extension. Furthermore, we
linked the naturally emerging problem to SAT. Finally, we demonstrated various
drawings in an experimental evaluation. The drawings produced by the algorithm
were, in our opinion, satisfying.  We would have liked to compare our algorithm
(exact and heuristic type) to the heuristics developed
in~\cite{conf/gd/KornaropoulosT12a}. Unfortunately, we were not able to
reproduce their results based on the provided description.

A notable observation is the fact that in all our experiments the SAT-Solver
blend of \texttt{DimDraw} was able to produce a solution in the first pass,
i.e., the algorithm found a truly minimal two-dimension extension. This raises
the natural question, whether the maximal induced bipartite subgraph approach
does always result in a minimal two-dimension extension. Further open questions
are concerned with employing heuristics and to improve the postprocessing
stage. The SAT-solver version of \texttt{DimDraw} is included in the software
conexp-clj\footnote{\url{https://github.com/tomhanika/conexp-clj}}. At a later
time we also want to include heuristic versions.

\subsection*{Acknowledgement}
The authors would like to thank Torsten Ueckerdt for pointing out the research
about diametral pairs and Maximilian Stubbemann for helpful discussions.

\bibliographystyle{bib/splncs04}
\bibliography{bib/paper}

\newpage
\appendix

\section{Additional Drawings}
\label{sec:additional-drawings}

In this section we provide some additional drawings generated by
\texttt{DimDraw}.

\begin{figure}
  \null\hfill
  \includegraphics[height=12em]{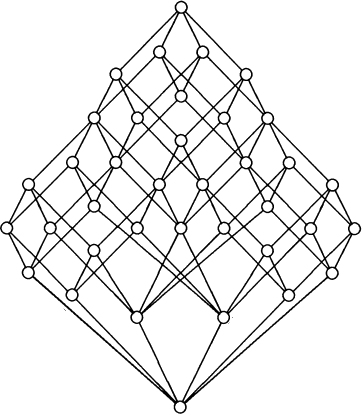}
  \hfill
  \includegraphics[height=12em]{konvex_ordinal}\hfill\null  
  \caption{An example from~\cite[p.53]{fca-book}. The hand-drawn version is on
    the left, the algorithm-generated version on the right.}
\end{figure}

\begin{figure}
  \null\hfill
  \includegraphics[height=20em]{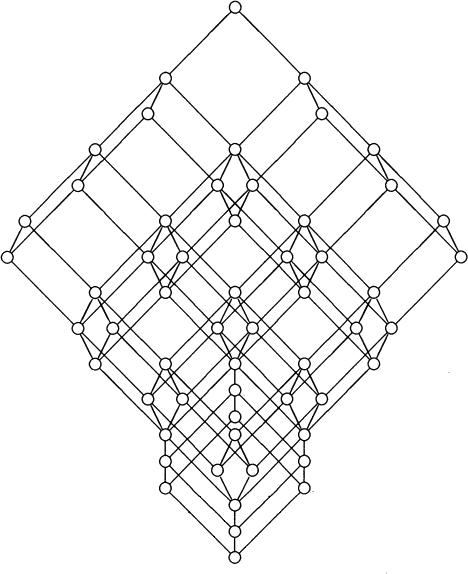}\hfill 
  \includegraphics[height=20em]{lattice} \hfill\null
  \caption{An example from~\cite[p.35]{fca-book}. The hand-drawn
    version is on the left, the algorithm-generated version on the right.}
\end{figure}

\begin{figure}
  \centering
  \begin{minipage}{0.6\linewidth}
    \includegraphics[height=10em]{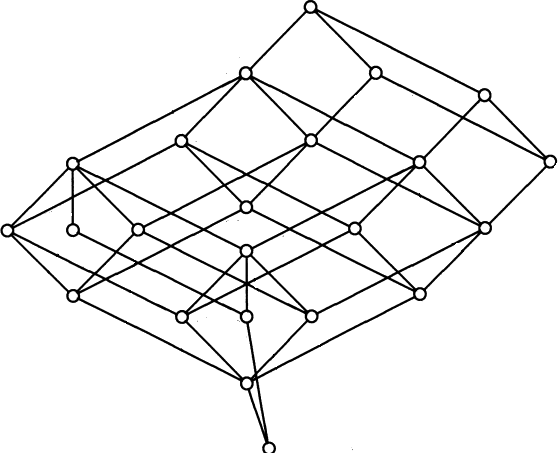}
    \includegraphics[height=10em]{dritte_welt}
    \caption{An example from~\cite[p.30]{fca-book}. The hand-drawn version is on
      the left, the algorithm-generated version on the right.}
  \end{minipage}
  \hfill
  \begin{minipage}{0.35\linewidth}
    \includegraphics[height=10em]{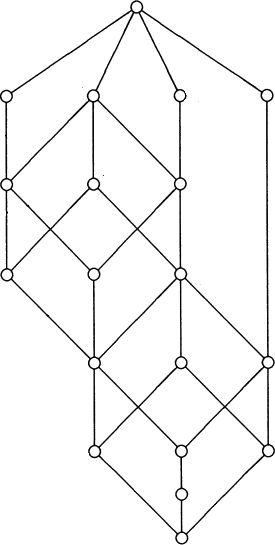}
    \includegraphics[height=10em]{forum}
    \caption{An example from~\cite[p.45]{fca-book}. The hand-drawn version is on
      the left, the algorithm-generated version on the right.}
  \end{minipage}

\end{figure}

\begin{figure}
  \null\hfill
  \includegraphics[height=10em]{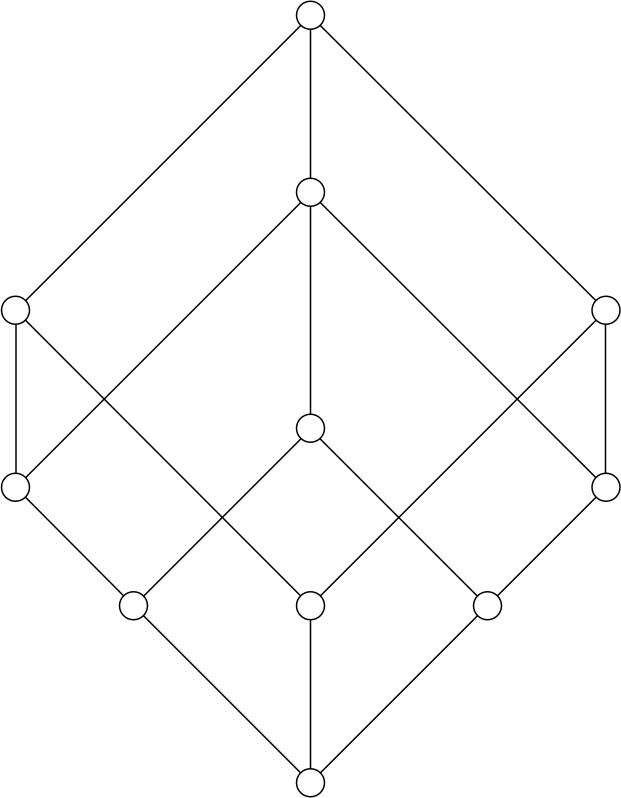}\hfill
  \includegraphics[height=10em]{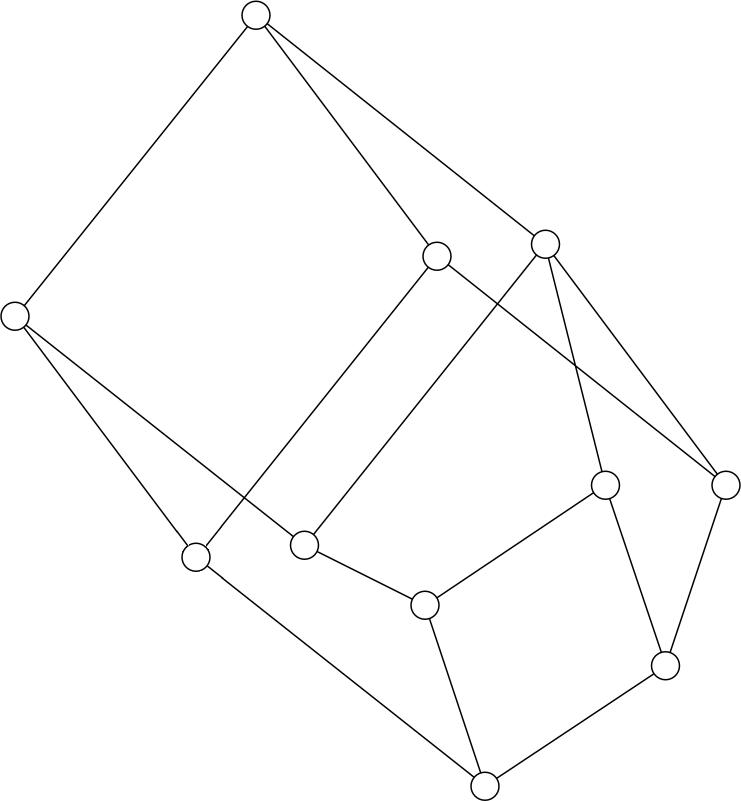}\hfill
  \includegraphics[height=10em]{ganter} \hfill\null
  \caption{Order diagram of the lattice used in \cite{ganteradd}. From left to
    right: Hand drawn, Ganter's algorithm, our algorithm.}
\end{figure}

\begin{figure}
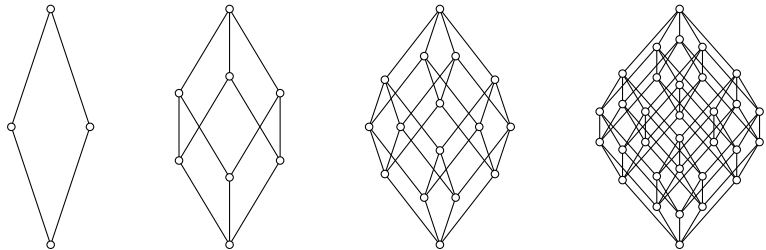

  \null\hfill
  \includegraphics[height=10em]{bool2}\hfill
  \includegraphics[height=10em]{bool3}\hfill
  \includegraphics[height=10em]{bool4}\hfill
  \includegraphics[height=10em]{bool5}\hfill\null
  \caption{Order diagrams for boolean lattices of dimension 2, 3, 4 and 5 drawn
    by our algorithm}
\end{figure}

\end{document}